\documentclass[manyauthors]{fundam}
\usepackage{url} 
\usepackage[ruled,lined]{algorithm2e}
\usepackage{graphicx}
\usepackage{hyperref}

\begin{document}

\setcounter{page}{321}
\publyear{2021}
\papernumber{2076}
\volume{182}
\issue{4}

  \finalVersionForARXIV

\title{Testing Boolean Functions Properties}

\author{Xie Zhengwei \\ 
  School of  Computer Science and Engineering \\
  Sun Yat-sen University,  Guangzhou 510006, China \\
   weixzh2010@163.com
\and
                 Qiu Daowen$^\dag$\thanks{Address for correspondence: School of Computer Sci. and Eng.,
                   Sun Yat-sen University,  Guangzhou 510006, China. \newline
                   $^\dag$ Also works: Instituto de Telecomunica\c{c}\~oes,
                       Dept. de Matem´atica, Instituto Superior T\'ecnico,  Lisbon, Portugal.  \newline  This work  is  partly supported
                         by the National Natural Science Foundation of China (Nos. 61572532,  61876195),
                         the Natural Science Foundation of Guangdong Province of China (No. 2017B030311011).
                         \newline \newline
          \vspace*{-6mm}{\scriptsize{Received  September 2021; \ revised November 2021.}}}
  \\
  School of  Computer Science and Engineering \\
  Sun Yat-sen Univ.,  Guangzhou 510006, China \\
  issqdw@mail.sysu.edu.cn
\and
Cai Guangya\\  
School of  Computer Science and Engineering \\
Sun Yat-sen University,  Guangzhou 510006, China \\
caigya@mail2.sysu.edu.cn
\and
Jozef Gruska\\ 
Faculty of Informatics\\
Masaryk University, Brno, Czech Republic\\
gruska@fi.muni.cz
  \and Paulo Mateus\\ 
Instituto de Telecomunica\c{c}\~oes\\
Dept. de Matem´atica, Instituto Superior T\'ecnico\\
Av. Rovisco Pais 1049-001 Lisbon, Portugal\\
pmat@math.ist.utl.pt
}

\maketitle

\runninghead{Z. Xie, et al.}{Testing Boolean Functions Properties}

\vspace*{-4mm}
\begin{abstract}
The goal in the area of functions property testing is to determine whether a given black-box Boolean function has a particular given property or is $\varepsilon$-far from having that property. We investigate here several types of properties testing for Boolean functions (identity, correlations and balancedness) using the Deutsch-Jozsa algorithm (for the Deutsch-Jozsa (D-J) problem) and also the amplitude amplification technique.

At first,  we study here a particular testing  problem: namely whether a given Boolean function $f$, of $n$ variables, is identical with a given function $g$ or is $\varepsilon$-far from $g$, where $\varepsilon$ is the parameter. We present a one-sided error  quantum algorithm to deal with this problem that has the query complexity $O(\frac{1}{\sqrt{\varepsilon}})$. Moreover, we show that  our quantum algorithm is optimal. Afterwards we show that the classical randomized query complexity of this problem is  $\Theta(\frac{1}{\varepsilon})$. Secondly, we consider the D-J problem from the perspective of functional correlations and let $C(f,g)$ denote the correlation of $f$ and $g$.  We propose an exact quantum algorithm for making distinction between $|C(f,g)|=\varepsilon$  and $|C(f,g)|=1$ using six queries, while the classical deterministic query complexity for this problem is $\Theta(2^{n})$ queries. Finally, we propose a one-sided error quantum query algorithm for testing whether one Boolean function is balanced versus $\varepsilon$-far balanced using $O(\frac{1}{\varepsilon})$ queries. We also prove here that our quantum algorithm for balancedness testing is optimal. At the same time, for this balancedness testing problem we present a classical randomized algorithm with query complexity of $O(1/\varepsilon^{2})$. Also this randomized algorithm is optimal. Besides, we link the problems considered here together and generalize them to the general case.
\end{abstract}


\section{Introduction}

We deal here with the function isomorphism testing problem in a new way. This problem was at first explored by Fischer {\it et al.}  \cite{JCSS04}. Function isomorphism testing is to determine whether a given Boolean function $f$ is isomorphic to a given specific function $g$ (This means that it is equal to $g$ up to a permutation of its input variables). So far only  a few simple testing problems have been already well understood. For example, it has been shown in \cite{SIAM15} that partially symmetric functions are isomorphism testable using a constant number of queries.

 Moreover, the following upper bound on the number of queries needed for isomorphism testing has been shown: for any integer $k$, if $g$ is a so-called $k$-junta (that is if $g$ depends on at most $k$ variables), then it is possible to test $g$-isomorphism with poly$(k/\varepsilon)$ number of queries. This bound was recently improved by Alon {\it et al.} \cite{SIAM13}. They showed that $O(k\log k)$  queries are sufficient for such testing. Recently, Chen {\it et al.}~\cite{JACM18} proved that any non-adaptive algorithm that tests whether an unknown Boolean function is a $k$-junta, or it is $\varepsilon$-far from every $k$-junta, must make $\tilde{\Omega}(k^{3/2}/\varepsilon )$ queries for a wide range of parameters $k$ and $\varepsilon$. A more generalize tolerant testing of $k$-junta was investigated in the reference \cite{SODA18}. Blais and O'Donnell \cite{CCC10}, Alon {\it et al.} \cite{SIAM13} proved more general lower bounds for the isomorphism testing. In particular, Blais {\it et al.} \cite{SIAM15} showed that all partially symmetric functions were efficiently isomorphism testable ({\it i.e.} - with a constant number of queries). In addition, there have been already also a lot of  results on other types of testing. For example, on the property testing such as linearity testing \cite{IEEE96}, monotonicity testing \cite{SIAM16,FOCS14}, group testing \cite{FI09} and so on.

Quantum query complexity \cite{QC1999,IC2020} is a well-known black-box model for quantum computation, in which the resource measured is the number of queries needed to compute a function. In particular, one provides a  ``black-box access" to a function $f$, meaning that the quantum algorithm can apply  some unitary transformation $U_{f}$ that maps basis states of the form $|x\rangle |y\rangle$ to basis states of the form $|x\rangle|y\oplus f(x)\rangle$ (or $|x\rangle$ to $(-1)^{f(x)}|x\rangle$, where $f$ is a Boolean function). This complexity model has been  behind very first great algorithmic successes of quantum computing like in the search algorithm of Grover \cite{STOC96} and in the algorithm for period finding which is a subroutine of Shor's  quantum factoring algorithm \cite{SIAM99}.

 A natural goal when such a model is considered is to minimize the number of queries to the oracle that are needed to solve the given problem. This minimum is the query complexity of the problem. However, quantum algorithms for property testing were much less studied than classical algorithms for this problem in the past. However, quantum property testing has been receiving increasing attentions in the last few years, both for testing properties of the classical objects and also for testing properties of quantum objects. For example, Buhrman {\it et al.}~\cite{SIAM08} showed that there exist languages with efficient quantum property testers but without so efficient classical testers. Hillery and Andersson \cite{PRA062329} presented two quantum algorithms for testing the linearity and the permutation invariance of Boolean functions. Moreover, Bravyi {\it et al.} \cite{IEEE11} described quantum algorithms for testing properties of distributions. More about property testing can be found in \cite{arXiv13}.

A more general concept, and also very interesting and important, than functions isomorphism is the affine equivalence of Boolean functions. Affine equivalence classification of Boolean functions has significant applications in logic synthesis and in cryptography. The affine equivalence classification of the cosets has been studied, for example, by  Hou \cite{DM94} and Zhang \cite{IEEE16}. It is easy to see that if $f$ and $g$ are affine equivalent, then they have the same weight and nonlinearity \cite{DCC12}. Canright {\it et al.} \cite{DM15} gave a simple necessary and sufficient condition for deciding whether two monomial rotation symmetric Boolean functions are affine equivalent. Fuller \cite{M03} made a detailed analysis of affine equivalent Boolean functions for cryptography.

Of a special interest are also the following results. Batu {\it et al.}\cite{FOCS01} considered the problem of testing distribution identity. Diakonikolas {\it et al.}~\cite{SODA15} studied the problem of identity testing for structured distributions. Moreover, in this paper, we focus on the identity testing for  Boolean functions (observe that the identity can be seen as a special case of the affine equivalence and isomorphism).

Moreover, Cai et al.\cite{MST96} computed the correlation between any two symmetric Boolean functions. They showed that every symmetric Boolean function having an odd period has an exponentially small correlation with the parity function. Castro et al.\cite{Ejc11} computed the asymptotic behavior of symmetric Boolean functions and derived a formula that allows to determine if a symmetric Boolean function is asymptotically not balanced. For example, for any integers $k_1, k_2$, ${\lim_{n \to \infty}} C(\sigma_{n,k_{1}},\sigma_{n,k_{2}})=0$, but $C(\sigma_{n,k_{1}},\sigma_{n,k_{2}})\neq0$, where each $\sigma_{n,k}$ is an elementary symmetric polynomial of $n$ variables and of the degree $k$ and $C(\sigma_{n,k_{1}},\sigma_{n,k_{2}})$ denotes  the correlation between this two functions. Castro et al.\cite{AC14} investigated the asymptotic behavior of  $ C(F_{1},F_{2})$, where $F_{1}= \sigma_{n,k_{1}} +...+\sigma_{n,k_{s}}$, $F_{2}=F(X_{1},..., X_{j})$ is a Boolean function in the first $j$-variables ($j$ fixed). Recently, Castro et al.\cite{AECC18} showed that the generalized Walsh transforms of symmetric and also rotation symmetric Boolean functions are linearly recurrent. This subject can been studied from the point of view of complexity theory or from the algebraic point of view. In this paper, we study such correlations from the perspective of property testing.

 Balancedness is one of  other important cryptographic properties of Boolean functions. There have been already many results, see for example \cite{IEEE16,DCC18} concerning the balancedness of Boolean functions in cryptography. They also inspired us to further study balancedness from the perspective of the property testing. In particular, we consider here  quantum and also classical randomized algorithms  for  balancedness testing.

The remainder of this paper is organized as follows. In Section 2, we present some  of the notations and definitions that will be used in the rest of the paper, as well as a review of the Deutsch-Jozsa algorithm and the quantum amplitude amplification theorem. Then, in Section 3, a quantum algorithms for testing Boolean functions identity is given and explored. Moreover, we provide the lower bound of testing Boolean functions identity. We further prove that the randomized query complexity of the problem  is  $\Theta(\frac{1}{\varepsilon})$. Also, we compare the complexity of testing identity of distributions with the complexity of testing identity of Boolean functions. In Section 4 and Section 5, we consider correlation testing problems and  balancedness testing problems, respectively. Finally, conclusions and open problems are summarized in Section 6.

\section{Preliminaries}\label{Sec2}

In this section, we first introduce some basic notions related to Boolean functions and quantum computation. For more details, see literature
\cite{QC1999,M03}.

Let $f:\{0,1\}^{n}\to\{0,1\}$ be a Boolean function of $n$ variables, and let ${B_{n}}$ denote the set of all Boolean functions of $n$ variables.
\begin{definition} (\textbf{Property tester}, \cite{SIAM16})
 Let $P_n$ be a set of Boolean functions of $n$ variables with a certain property. An $\varepsilon$-tester, $\varepsilon>0$, for $P_n$ is a randomized or quantum algorithm which can query an unknown Boolean function $f:\{0,1\}^{n}\rightarrow\{0,1\}$ on a small number of inputs and
\begin{itemize}
  \item [(1)] accepts $f$ with the probability at least $\frac{2}{3}$ when $f \in P_n$,
  \item [(2)] rejects $f$ with the probability at least $\frac{2}{3}$ when $f$ is $\varepsilon$-far from any function in $P_n$,
\end{itemize}
 where ``$f$ is $\varepsilon$-far from $P_n$" means that  $\mbox{Ham}(f,g)=|\{x\in\{0,1\}^{n}:f(x)\neq g(x)\}|\geq\varepsilon2^{n}$ holds for every $g \in P_n$. We also denote by $\mbox{Dist}(f,g)=\mbox{Ham}(f,g)/2^{n}$.
\end{definition}
\begin{definition}
When $x$ is a string, let $|x|$ denote the Hamming weight of the string $x$, i.e., the number of ones in the string $x$.
\end{definition}
\begin{definition}
Let $x=(x_{1},\ldots,x_{n})$ and $\omega=(\omega_{1},\ldots,\omega_{n})$ both belong to $\{0,1\}^{n}$. The inner product of $x$ and $\omega$ is defined as $x \cdot \omega=x_{1}\omega_{1}\oplus \ldots \oplus x_{n}\omega_{n}$, \mbox{where}
$\oplus$ is addition modulo 2.
\end{definition}
\begin{definition}\label{Walsh}
 Suppose $f\in{B_{n}}$. The Walsh transform of $f$, is defined as
 \begin{displaymath}
W_{f}(\omega)=\frac{1}{2^{n}}\sum_{x\in\{0,1\}^{n}}(-1)^{f(x)+\omega \cdot x}, \mbox{where} \hspace{0.15cm}\omega\in\{0,1\}^{n}.
\end{displaymath}
\end{definition}
\begin{definition}( \cite{AC14} )\label{correlation}
The correlation $C(f,g)$ between two Boolean functions $f$ and $g$ of $n$ variables is defined as the number of arguments for which functions $f$ and $g$ have the same value, minus the number of arguments  on which they disagree, all divided by $2^{n}$, i.e.,
\begin{displaymath}
 C(f,g)=\frac{1}{2^{n}}\sum_{x\in\{0,1\}^{n}}{(-1)}^{f(x)+g(x)}.
\end{displaymath}
\end{definition}For the case of one function $f$, let's denote $ C(f)=\frac{1}{2^{n}}\sum_{x\in\{0,1\}^{n}}{(-1)}^{f(x)}$. A Boolean function $f(x)$ is called balanced if $C(f)=0$.
\begin{remark}
According to Definition \ref{Walsh} and \ref{correlation}, we have $ C(f)=W_{f}(0)$.
\end{remark}
\begin{definition}\label{isomorphic}
(\cite{CCC10}) Two Boolean functions $f,g\in{B_{n}}$ are said to be isomorphic if they are identical up to any reordering of input variables. More precisely, we say that $f,g\in{B_{n}}$ are isomorphic (to each other) if there is a permutation $\sigma$ on $[n]$ such that for every $x=(x_{1},x_{2},\ldots,x_{n})\in \{0,1\}^{n}$, $f(x_{1},x_{2},\ldots,x_{n})=g(x_{\sigma_{1}},x_{\sigma_{2}},\ldots,x_{\sigma_{n}})$, where $\sigma(1,2,\ldots,n)=(\sigma_1,\sigma_2,\ldots,\sigma_n)$.
\end{definition}

\begin{definition}\label{affine}
(\cite{IEEEC16}) Two Boolean functions  $f(x_{1},x_{2},\ldots,x_{n})$ and $h(x_{1},x_{2},\ldots,x_{n})$  are said to be affine equivalent if $h(x_{1},x_{2},\ldots,x_{n})$ can be written as $h(x_{1},x_{2},\ldots,x_{n})=h(X)=f(AX+b)$, where $A$ is an $n\times n$ nonsingular matrix over the Boolean field ${F_{2}}$, ${X}$ is a column vector whose transpose is $X^{T}=[x_{1},x_{2},\ldots,x_{n}]$, and $b$ is an $n$-dimensional vector over ${F_{2}}$. The addition $AX+b$ is also over ${F_{2}}$.
\end{definition}

Affine equivalent functions have many similar cryptographic properties, such as correlation immunity, resiliency and propagation characteristics (For example, see \cite{IEEEC16}). Moreover, Zhang et al.\cite{IEEEC16} presented a new method for computing affine equivalence classes of Boolean functions due to a  group isomorphism. In the following proposition, we establish a connection between the isomorphism relation  and the affine equivalence relation of two Boolean functions.

\begin{proposition}
 Let $f,g\in{B_{n}}$. Then $f$ is isomorphic to $g$ if and only if there exists $A=\left(
                                                                                                           \begin{array}{c}
                                                                                                             e_{\sigma(1)} \\
                                                                                                             e_{\sigma(2)} \\
                                                                                                             \vdots \\
                                                                                                             e_{\sigma(n)} \\
                                                                                                           \end{array}
                                                                                                         \right)
$, $b=\textbf{0}$ such that $f(X)=g(AX+b)$, where $e_{1}=(1,0,\ldots,0)$,\ldots,$e_{n}=(0,0,\ldots,1)$ and $\sigma$ is a permutation on $[n]$.
\end{proposition}
\begin{proof}
  This follows easily from Definition \ref{isomorphic} and Definition \ref{affine}.
\end{proof}
\begin{definition}
(\cite{PRA062331})  Let $f:\{0, 1\}^n \rightarrow \{0, 1\}$ be a partial Boolean function, and  $D\subseteq \{0, 1\}^n$ be its domain of definition. If for any $x\in D$ and any $y\in \{0, 1\}^n$ with $|x|=|y|$, it holds that $y\in D$ and $f(x)=f(y)$, then $f$ is called a {\em  symmetric  partial Boolean function}. When $D=\{0,1\}^n$, $f$ is a symmetric function.
\end{definition}
\begin{definition}
(\cite{IC2020,PRA062331}) Let $f$ be a partial Boolean function with a domain of definition $D\subseteq\{0,1\}^n$.
 We say a real multilinear polynomial $p$ approximates $f$ with the error $1/3$ if:
\begin{enumerate}
\item[(1)]  $|p(x)-f(x)|\leq 1/3$ for all $x\in D$;
\item[(2)]  $0\leq p(x)\leq 1$ for all $x\in\{0,1\}^n$.
\end{enumerate}
The approximate degree of $f$ with the error $1/3$, denoted by $\widetilde{\text{deg}}(f)$, is the minimum degree among all real multilinear polynomials  that approximate $f$ with the error $1/3$.
\end{definition}

\begin{lemma}\label{Lemma-degree-II}
(\cite{IC2020,PRA062331}) Let $f$ be a symmetric  partial  Boolean  function  over $\{0,1\}^n$ with the domain of definition $D$,
 and $\widetilde{\text{deg}}(f)= d$. Then there exists a real multilinear polynomial $q$  that approximates $f$ with the error $1/3$ and $q$ can be written as
\begin{align}
 q(x)=c_0+c_1V_1+c_2V_2+\cdots+c_dV_d,
\end{align}
where $c_i\in{R}$, $V_j$ denotes the sum of all products of $j$ different variables, i.e.,$V_1=x_1+\cdots+x_n$, $V_2=x_1x_2+x_1x_3+\cdots+ x_{n-1}x_n$, etc.
\end{lemma}
\begin{remark}Note that $V_{j}$ assumes value $\left(
  \begin{array}{c}
    |x| \\
    j \\
  \end{array}
\right)$$=|x|(|x|-1)(|x|-2)\cdots(|x|-j+1)/j!$ for any string $x$, and $\left(
  \begin{array}{c}
    |x| \\
    j \\
  \end{array}
\right)$ is a polynomial of degree $j$ of $|x|$.
 Hence  $\phi(|x|)=$ $c_0$+$c_1$ $\left(
  \begin{array}{c}
    |x| \\
    1 \\
  \end{array}
\right)$+$c_2$ $\left(
  \begin{array}{c}
    |x| \\
    2 \\
  \end{array}
\right)$+$\cdots$+$c_d$$\left(
  \begin{array}{c}
    |x| \\
    j \\
  \end{array}
\right)$$= q(x)$. Therefore we can obtain a single-variate polynomial $\phi(|x|)$ which approximates $f$ with the error $1/3$ such that $\text{deg}(\phi)=\widetilde{\text{deg}}(f)$.
\end{remark}
\begin{lemma}\label{Lemma-lower-bound1}
(\cite{MZ64,TCS02}) Let $p$ be a polynomial with the following properties:
\begin{enumerate}
\item  For any integer $0\leq i\leq n$, we have $b_{1}\leq p(i)\leq b_{2}$ for some fixed $b_1$ and $b_2$.
\item  For some real $0\leq x\leq n$, the derivative $p'(x)$ of $p$ satisfies $\mid p'(x)\mid\geq c$ for some $c$.
\end{enumerate}
Then, $ \mbox{deg}(p)\geq \sqrt{cn/(c+b_{2}-b_{1})}$.
\end{lemma}
\begin{lemma}\label{Lemma-lower-bound2}
(\cite{PRA062331}) For any partial Boolean function $f$, $\widetilde Q(f)\geq \frac{1}{2}\widetilde{\text{deg}}(f)$, where $\widetilde Q(f)$ denotes the quantum query complexity for $f$ with the bounded-error $1/3$.
\end{lemma}
We review  now briefly the Deutsch-Jozsa algorithm \cite{PRS92} task and the quantum amplitude amplification \cite{CM02}. The Deutsch-Jozsa algorithm is to distinguish for a given Boolean function $f$ one of two cases with certainty using only one query, under the promise that the function $f$ is either constant or balanced. (Here, we do not describe details of the D-J algorithm - they will not be of importance in the following.)

\medskip
Let $\mathcal{A}$ be any quantum algorithm (a unitary operator) that acts on a Hilbert space $\mathcal H$ and uses no measurement. Let $|\Psi\rangle =\mathcal{A}|0\rangle$ denote the state obtained by applying the unitary $\mathcal{A}$ to the initial zero state. The amplification process is realized by repeating an application of the following unitary operator on the state $|\Psi\rangle$,
\begin{align}\label{1}
Q=-\mathcal {A}S_{0}\mathcal{A}^{-1}S_{\chi},
\end{align}
where the operator $S_{\chi}$ conditionally changes the phase of so called good states,
\begin{displaymath}
|x\rangle\rightarrow \left\{ \begin{array}{ll}
-|x\rangle & \textrm{if $\chi(x)=1$},\\
\hspace{0.25cm}|x\rangle & \textrm{if $\chi(x)=0,$}
\end{array} \right.
\end{displaymath}
while the operator $S_{0}$ changes the sign of the amplitudes if and only if the state is the zero state $|0\rangle$. Write $|\Psi\rangle=|\Psi_{1}\rangle+|\Psi_{0}\rangle$ is a superposition of so called good and so called bad components (basis substates) of $|\Psi\rangle$. We have that
\begin{align}\label{1}
Q|\Psi_{1}\rangle&=(1-2a)|\Psi_{1}\rangle-2a|\Psi_{0}\rangle,\\[3pt]
Q|\Psi_{0}\rangle&=2(1-a)|\Psi_{1}\rangle+(1-2a)|\Psi_{0}\rangle,\\
Q^{j}|\Psi\rangle&=\frac{1}{\sqrt{a}}\sin((2j+1)\theta_{a})|\Psi_{1}\rangle+\frac{1}{\sqrt{1-a}}\cos((2j+1)\theta_{a})|\Psi_{0}\rangle,
\end{align}where $\sin^{2}(\theta_{a})=a=\langle\Psi_{1}|\Psi_{1}\rangle$.

\begin{lemma}\label{lm1}
 (\textbf{Quadratic speedup} \cite{CM02}) Let $\mathcal {A}$ be any quantum algorithm that uses no measurements, and let $\chi:Z\rightarrow\{0,1\}$ be any Boolean function. Let $a$ be the initial success probability of $\mathcal {A}$. Suppose $a>0$, and set $m=\lfloor {\pi/4\theta_{a}}\rfloor$, where $\theta_{a}$ is defined such that $\sin^{2}(\theta_{a})=a$ and $0<\theta_{a}\leq\pi/2$. Then, if we compute $Q^{m}\mathcal{A}|0\rangle$ and measure the system, the outcome is good with probability at least $\max(1-a,a)$.
\end{lemma}

When the success probability $a$ of quantum algorithm $\mathcal{A}$ is known, we redefine the unitary operator $Q$ to find a good solution with certainty. Let
\begin{align}\label{5}
Q=Q(\mathcal {A},\phi,\varphi,\chi)=-\mathcal {A}S_{0}(\phi)\mathcal{A}^{-1}S_{\chi}(\varphi).
\end{align}
Here the operator $S_{\chi}(\varphi)$ conditionally changes the phase of good states,
\begin{displaymath}
|x\rangle\rightarrow \left\{ \begin{array}{ll}
e^{i\varphi}|x\rangle & \textrm{if $\chi(x)=1$}\\
\hspace{0.45cm}|x\rangle & \textrm{if $\chi(x)=0,$}
\end{array} \right.
\end{displaymath}
while the operator $S_{0}(\phi)$ multiplies the amplitude by a factor of $e^{i\phi}$ iff the initial state is the zero state~$|0\rangle$.
\begin{lemma} \label{Generalized Amplitude Expansion}
\cite{CM02} Let $Q=Q(\mathcal {A},\phi,\varphi,\chi)$ be defined as above, then
\begin{align}\label{6}
Q|\Psi_{1}\rangle &=e^{i\varphi}((1-e^{i\phi})a-1)|\Psi_{1}\rangle +e^{i\varphi}(1-e^{i\phi})a|\Psi_{0}\rangle\\[2pt]
Q|\Psi_{0}\rangle &=(1-e^{i\phi})(1-a)|\Psi_{1}\rangle -((1-e^{i\phi})a+e^{i\phi})|\Psi_{0}\rangle,
\end{align}
where $a=\langle \Psi_{1}|\Psi_{1}\rangle $ .
\end{lemma}

\section{Testing identity of two unknown Boolean functions}\label{Sec4}

 Testing of the function isomorphism with classical or quantum algorithms has already been investigated in the literature \cite{CCC10,SIAM15}. In spite of the fact that the affine equivalence has been already studied in cryptography, there are only few results concerning  quantum algorithms for  affine equivalence testing. Proposition 1 establishes the relation between the affine equivalence and the isomorphism, and this may contribute to the further study of the  testing of the affine equivalence. Identity can be seen as a special case of isomorphism and affine equivalence. Motivated by the isomorphism testing, we consider the identity testing problem. In this section, we give a quantum algorithm for testing  Boolean functions identity. Moreover, our algorithm will be shown  to be optimal. Then we present a lower bound on classical algorithms for the same problem. We also show that such lower bound is \linebreak reachable.

\subsection{Quantum lower and upper bounds for  testing  Boolean functions identity}

\textbf{Problem 1} The Identity Testing Problem:
 Given two unknown Boolean functions $f,g$, and an $\varepsilon> 0$, the testing  problem is to determine whether $f$ is identical with  $g$ or is $\varepsilon$-far from $g$, under the promise that one of them holds.

\medskip
 In the black box model, one provides a  ``black-box access" to a function $f$, meaning that the quantum algorithm can apply some unitary transformation
$U_{f}:|x\rangle |y\rangle\rightarrow |x\rangle|y\oplus f(x)\rangle$ (or $|x\rangle$ to $(-1)^{f(x)}|x\rangle$, where $f$ is a Boolean function). We present now a quantum algorithm for the above testing problem: Algorithm \ref{AL-QIT}.\smallskip

\begin{algorithm}
\caption{Quantum Algorithm for Identity Testing}\label{AL-QIT}\smallskip
\KwIn{Black-boxes for $f$ and $g$}
\KwOut{ $f$ is identical with $g$ iff $z=0$}
$|\psi_{0}\rangle\leftarrow {|{0}\rangle^{\otimes n}\otimes|1\rangle}$\;
Apply ${I^{\otimes n}\otimes  H}$ to $|\psi_{0}\rangle$ and get
$|\psi_{1}\rangle=|{0}\rangle^{\otimes n}\otimes|-\rangle$\;
Apply ${H^{\otimes n}\otimes  I}$ to $|\psi_{1}\rangle$ and get
$|\psi_{2}\rangle=\sum_{x\in\{0,1\}^{n}}\frac{1}{\sqrt{2^{n}}}|x\rangle|-\rangle$\;
Apply ${U_{f}}$ to $|\psi_{2}\rangle$ and get $|\psi_{3}\rangle=\sum_{x\in\{0,1\}^{n}}\frac{{(-1)}^{f(x)}}{\sqrt{2^{n}}}|x\rangle|-\rangle$\;
Apply ${U_{g}}$ to $|\psi_{3}\rangle$ and get $|\psi_{4}\rangle=\sum_{x\in\{0,1\}^{n}}\frac{{(-1)}^{f(x)+g(x)}}{\sqrt{2^{n}}}|x\rangle|-\rangle$\;
Apply ${H^{\otimes n}\otimes  I}$ to $|\psi_{4}\rangle$ and get $|\psi_{5}\rangle= \sum_{z\in\{0,1\}^{n}}{W_{h}(z)}|z\rangle|-\rangle$, where $h=f+g$\;	
Apply $Q^{m}$ to $|\psi_{5}\rangle$ and get $|\psi_{final}\rangle$, where $Q=-\mathcal {A}S_{0}\mathcal{A}^{-1}S_{\chi}$, $m=O(\frac{1}{\sqrt{\varepsilon}})$, $\mathcal {A}|{0}\rangle^{\otimes n}|-\rangle=\sum_{z\in\{0,1\}^{n}}{W_{h}(z)}|z\rangle|-\rangle$, $\mathcal {A}=(H^{\otimes n}\otimes  I)({U_{g}})({U_{f}})(H^{\otimes n}\otimes  I)$ and $|z\rangle|-\rangle(z\neq0)$ are good states\;
Measurement the first $n$ qubits and get $z$\;
Return $z$\; \smallskip
\end{algorithm}

\begin{theorem}\label{upbound1}
Algorithm \ref{AL-QIT} solves Identity Testing Problem 1 with one-sided error $\varepsilon>0$ using $O(\frac{1}{\sqrt{\varepsilon}})$ queries. If $f$ is identical with $g$, Algorithm \ref{AL-QIT} outputs ``$f$ is identical with $g$" with certainty. If $f$ is $\varepsilon$-far from $g$, Algorithm  \ref{AL-QIT} outputs ``$f$ is $\varepsilon$-far from $g$" with the probability at least $\frac{3}{4}$.
\end{theorem}

\eject
\begin{proof}
 According to Definition \ref{Walsh}, we obtain
\begin{align}\label{3}
W_{h}(0)=&\sum_{x\in\{0,1\}^{n}}\frac{{(-1)}^{f(x)+g(x)}}{2^{n}}\\
  =&\frac{1}{2^{n}}\sum_{x:f(x)=g(x)}{{(-1)}^{f(x)+g(x)}}+\frac{1}{2^{n}}\sum_{x:f(x)\neq g(x)}{(-1)}^{f(x)+g(x)}\\
 =&\frac{1}{2^{n}}[2^{n}-2\cdot\mbox{Ham}(f,g)]\\
 =&1-2\cdot\mbox{Dist}(f,g).
\end{align}
 Without loss of generality, we assume $\varepsilon\leq \mbox{Dist}(f,g)\leq 1/15 $. Otherwise, Problem 1 is easy to deal with classically.

\medskip
Case 1: $f$ is identical with $g$. Then the final state $|\psi_{final}\rangle$ in Algorithm \ref{AL-QIT} is ${|{0}\rangle^{\otimes n}\otimes|-\rangle}$. Therefore, Algorithm \ref{AL-QIT} outputs ``$f$ is identical with $g$" with a certainty.

\medskip
Case 2: $f$ is $\varepsilon$-far from $g$, $0<W_{h}(0)=1-2\cdot \mbox{Dist}(f,g)\leq 1-2\varepsilon$. Then the final state $|\psi_{final}\rangle$ in Algorithm \ref{AL-QIT} is:
\begin{align}\label{2}
\sum_{z\in\{0,1\}^{n}}{W_{h}(z)}|z\rangle|-\rangle=&|\phi_{1}\rangle|-\rangle+|\phi_{0}\rangle|-\rangle, \mbox{where}\\
|\phi_{1}\rangle=&\sum_{z\neq 0}{W_{h}(z)}|z\rangle,\\
|\phi_{0}\rangle=&\sum_{z=0}{W_{h}(z)}|z\rangle,
\end{align} and therefore
\begin{align}\label{3}
a &=\sin^{2}\theta_{a}=\langle\phi_{1}\mid\phi_{1}\rangle\\
  &=\sum_{z\neq 0}{W^{2}_{h}(z)}=1-\sum_{z= 0}{W^{2}_{h}(z)}\\
 &\geq1-(1-2\varepsilon)^{2}=4\varepsilon-4\varepsilon^{2}\geq\frac{56}{15}\varepsilon.
\end{align}
 Therefore, we obtain $\theta_{a}\geq \arcsin\sqrt{\frac{56}{15}\varepsilon}$. To prepare the final state $|\psi_{final}\rangle$ in Algorithm \ref{AL-QIT}, we need to call black box two times, i.e., that is to use two queries. So it takes 4 queries (i.e., $\mathcal{A}$ uses two queries, $\mathcal{A}^{-1}$ uses two queries too) when Algorithm \ref{AL-QIT} runs $Q$ once time, where $Q=-\mathcal {A}S_{0}\mathcal{A}^{-1}S_{\chi}$. According to Lemma \ref{lm1}, Algorithm \ref{AL-QIT} runs $Q$ $m$ times, where $m=\lfloor {\pi/4\theta_{a}}\rfloor$. Finally, we know that the total number of queries of Algorithm \ref{AL-QIT} is $4m+2=O(\frac{1}{\sqrt{\varepsilon}})$. Hence the query complexity of Algorithm \ref{AL-QIT} is $O(\frac{1}{\sqrt{\varepsilon}})$. According to the hypothesis $\varepsilon\leq \mbox{Dist}(f,g)\leq 1/15 $, we have $(\frac{13}{15})^{2}\leq W^{2}_{h}(0)\leq(1-2\varepsilon)^{2}$, {\it i.e.}, $a\leq56/225$. Again, by Lemma \ref{lm1}, Algorithm \ref{AL-QIT} outputs ``$f$ is $\varepsilon$-far from $g$" with the probability $p\geq \rm max(1-56/225,56/225)>\frac{3}{4}$.
\end{proof}

\begin{remark}
 For an unknown function $f$, we can test $f$ using different types of $g$ so as to obtain an information concerning $f$. Also we can test two unknown functions whether they are identical or they are $\varepsilon$-far from each other. Besides, in the following quantum algorithms $|-\rangle$ also stands for $\frac{|0\rangle-|1\rangle}{\sqrt{2}}$.
\end{remark}
As the next we show that the quantum lower bound for Problem 1 is also $\Omega(\frac{1}{\sqrt{\varepsilon}})$. That means that Algorithm 1 is optimal.
\begin{theorem}\label{lbound1}
Any quantum query algorithm for Identity Testing Problem 1 requires $\Omega(\frac{1}{\sqrt{\varepsilon}})$ queries.
\end{theorem}
\begin{proof}
 Let $X=(x_{0},x_{1},\ldots,x_{N-1})$, $Y=(y_{0},y_{1},\ldots,y_{N-1})$. Let $f(i)=x_{i}$, $g(i)=y_{i}$, for $i=0,1,\cdots,N-1$, $x_{i},y_{i}\in \{0,1\}$ and $N=2^{n}$. Values of any $n$-bit Boolean function can now be expressed by an $N$-bit string. The task to solve Problem 1 can now be reduced to that of computing the Hamming weight of an $N$-bit string $Z=(x_{0}\oplus y_{0},x_{1}\oplus y_{1},\ldots,x_{N-1}\oplus y_{N-1})$.

\medskip
In order to do that, we consider any $N$-bit symmetric partial  Boolean function $H: \{D\subset\{0,1\}^{N}\}\to\{0,1\}$ defined as follows:
\begin{align}
	H(Z)= \begin{cases}
		   0 & \mbox{if~} |Z|=0, \\
           1 & \mbox{if~} |Z|\geq \varepsilon N, \\
           \mbox {undefined} & \mbox{otherwise.}
	\end{cases}
\end{align}
 Solving Problem 1 is now equivalent to computing symmetric partial Boolean function $H(Z)$. By Lemma \ref{Lemma-degree-II}, there exists an univariate polynomial $\phi(t)$ approximating $H$ with an error $1/3$ and $\widetilde{\text{deg}}(H)\!=\! deg(\phi)$. According to the Lagrange's mean value theorem on the interval $[0,\varepsilon N],\,$ we$\,$have
\begin{align}
	\phi(\varepsilon N)-\phi(0)=\phi'(\xi)(\varepsilon N),0< \xi<\varepsilon N.
\end{align}
Since the polynomial $\phi(t)$ approximates $H$ with the error $1/3$, we have
\begin{gather}
	0\leq\phi(0)\leq 1/3,\hspace{0.15cm} 2/3\leq \phi(\varepsilon N)\leq1,\\
        \phi(\varepsilon N)-\phi(0)\geq1/3>0.
\end{gather}
 Therefore, $\phi'(\xi)=c_{0}/\varepsilon N$, where $c_{0}=\phi(\varepsilon N)-\phi(0)$. In other words, $\phi(t)$ has the following properties:
\begin{enumerate}
\itemsep=0.9pt
\item  $0\le \phi(i)\le 1$ for any integer $0\leq i\leq N$.
\item  There exists $\xi$ such that $\mid \phi'(\xi)\mid\geq c_{0}/\varepsilon N$.
\end{enumerate}
  By Lemma \ref{Lemma-lower-bound1}, we obtain
  \begin{align}\label{1}
\widetilde{\text{deg}}(H)&=deg(\phi)\\
 &\geq \sqrt{(c_{0}N/\varepsilon N) /(c_{0}/\varepsilon N+1-0)}\\
 &= \sqrt{c_{0} /(c_{0}/ N+\varepsilon)}\\
 &=\Omega(\frac{1}{\sqrt{\varepsilon}}).
\end{align}
Finally, by Lemma \ref{Lemma-lower-bound2}, we get $\widetilde Q(H)=\Omega(\frac{1}{\sqrt{\varepsilon}})$.
\end{proof}

\vfil\eject

\begin{remark}
According to Theorem \ref{upbound1} and Theorem \ref{lbound1}, the quantum query complexity of Problem 1 is $\Theta(\frac{1}{\sqrt{\varepsilon}})$.
\end{remark}

\subsection{Classical lower and upper bounds  for testing Boolean functions identity}

It is now natural to ask whether our quantum algorithm has any advantage comparing  to its classical counterparts.In order to do that, we present in this subsection  a lower bound for randomized classical algorithms for testing Boolean functions identity.

In order to prove such a lower bound, we will use the following definition.
\begin{definition}
We say that inputs ${X_{i_{1}}},{X_{i_{2}}},\cdots,{X_{i_{k}}}$ are good for some functions $f$ and $g$ if there exists ${x_{i_{j}}}$ such that $f({X_{i_{j}}})\neq g({X_{i_{j}}})$, where ${X_{i_{l}}}\in\{0,1\}^{n}$ for $l=1,2,\ldots,k$. Otherwise, the inputs ${X_{i_{1}}},{X_{i_{2}}},\cdots,{X_{i_{k}}}$ will said to be bad. We will denote the event \\
\begin{align}\notag
B_{k}=\{X_{i_{1}}, {X_{i_{2}}},\cdots,{X_{i_{k}}}| \hspace{0.05cm}\mbox{inputs }\hspace{0.10cm} {X_{i_{j}}} \hspace{0.05cm}\mbox{are}\hspace{0.15cm} \mbox{bad},1\leq j\leq k\}.
\end{align}
\end{definition}

\begin{theorem}\label{lbound2}
Any classical randomized algorithm for Identity Testing Problem 1 requires $\Omega(\frac{1}{\varepsilon})$ queries.
\end{theorem}
\begin{proof}
Let $\mathcal {A}$ be a randomized algorithm for Problem 1 and $N=2^{n}$. Let us also suppose that $\mathcal {A}$ makes some random adaptive queries for $m$ different inputs $X_{i_{1}},X_{i_{2}},\cdots,X_{i_{m}}$. If $f$ is $\varepsilon$-far from $g$, according to the conditional probability formula and the fact that $B_{m}\subseteq B_{m-1}\subseteq \cdots \subseteq B_{1}$, we get that
\begin{align}\label{1}
P(B_{m})=&P(B_{m}|B_{m-1})P(B_{m-1})\\
=&P(B_{m}|B_{m-1})P(B_{m-1}|B_{m-2})P(B_{m-2})\\
 =&P(B_{1})\Pi_{k=2}^{m}P(B_{k}|B_{k-1})\\
 =&(1-\varepsilon)\frac{(1-\varepsilon)N-1}{N-1}\frac{(1-\varepsilon)N-2}{N-2}\cdots\frac{(1-\varepsilon)N-(m-1)}{N-(m-1)}\\
 =&(1-\varepsilon)\left[1-\frac{\varepsilon N}{N-1}\right]\left[1-\frac{\varepsilon N}{N-2}\right]\cdots\left[1-\frac{\varepsilon N}{N-(m-1)}\right]\\
\geq&(1-\varepsilon)\left[1-\sum_{k=1}^{m-1}\frac{\varepsilon N}{N-k}\right]\\
\geq&(1-\varepsilon)\left[1-\frac{\varepsilon (m-1)N}{N-(m-1)}\right]\\
\geq&(1-\varepsilon)\left[1-\frac{\varepsilon mN}{N-(m-1)}\right].
\end{align}
If $m=o(\frac{1}{\varepsilon})$ and $\varepsilon$ is sufficiently small, then $(1-\varepsilon)\left[1-\frac{\varepsilon mN}{N-(m-1)}\right]\approx1$. That is $P(B_{m})\approx1$. Therefore, the algorithm cannot determine whether $f$ is identical with $g$ or is $\varepsilon$-far from $g$ with a high success probability.
\end{proof}

The above lower bound is tight. Indeed, we give now an algorithm to reach the lower bound.
\begin{algorithm}\small
\caption{Classical Randomized Algorithm for Identity Testing}\label{ClaIdTest}
\KwIn{Black-boxes for $f$ and $g$, parameter $\varepsilon>0$}
\KwOut{$f$ is $\varepsilon$-far from $g$ iff $f(x)\neq g(x)$ for some $x$ }
$r\leftarrow 0$\;
\While{$r\leq \frac{\ln3}{\varepsilon}$}{Take an element $x\in\{0,1\}^{n}$ uniformly at random\;
\If{$f(x)\neq g(x)$}{outputs ``$f$ is $\varepsilon$-far from $g$" and halts\;}
$r\leftarrow r+1$\;}
\end{algorithm}\vspace*{-2mm}

\begin{theorem}\label{upbound2}
 Algorithm \ref{ClaIdTest} solves Problem 1 with one-sided error using $O(\frac{1}{\varepsilon})$ queries. If $f$ is identical with  $g$, our Algorithm \ref{ClaIdTest} outputs ``$f$ is identical with $g$" with certainty. If $f$ is $\varepsilon$-far from $g$, Algorithm \ref{ClaIdTest} outputs ``$f$ is $\varepsilon$-far from $g$" with the probability at least $\frac{2}{3}$.
\end{theorem}
\begin{proof}
 Case 1: Let $f$ be identical with  $g$. Algorithm \ref{ClaIdTest} then clearly outputs ``$f$ is identical with $g$" with certainty.\medskip\\
Case 2: Let $f$ be $\varepsilon$-far from $g$. If Algorithm \ref{ClaIdTest} runs $r$ loops, Algorithm \ref{ClaIdTest}  reports $f$ is identical with  $g$ with an error probability $p_{r}\leq(1-\varepsilon)^{r}< \frac{1}{3}$. The above inequality then holds when we take $r>\frac{\ln3}{\varepsilon}$, \it i.e., $r=O(\frac{1}{\varepsilon})$.
\end{proof}

\begin{remark}
According to Theorem \ref{lbound2} and Theorem \ref{upbound2}, the randomized query complexity for our Problem 1 is $\Theta(\frac{1}{\varepsilon})$. Observe also that Batu {\it et al.} \cite{FOCS01} and Diakonikolas {\it et al.}~\cite{SODA15} studied the question of the identity testing for  all kinds of distributions. The query complexity of these algorithms is not only related to the metric distance parameter $\varepsilon$, but also to the distribution size.  The query complexity of our algorithm depends on the value of  $\varepsilon$ and not on input size of Boolean functions. The reasons for this need to be further discussed.
\end{remark}

\section{Testing  correlations between two unknown Boolean functions }\label{Sec5}

 Various properties and correlations of some and between symmetric Boolean functions have been already explored extensively including some asymptotic properties. For example, the Deutsch-Jozsa quantum algorithm solved Deutsch's problem using one query: it is the problem (as already disussed) to decide whether a given function  $f$ is constant or balanced under the promise that it has one of these two properties. This problem is equivalent to the problem  to distinguish $|x|=2^{n-1}$ from $|y|=0,{2^{n}}$, where $|y|$ denotes the Hamming weight of string $y$. Gruska, Qiu and Zheng \cite{MSSC2017} generalized the distributed Deutsch-Jozsa promise problem. Recently, Qiu and Zheng \cite{PRA062331} also generalized the D-J problem and gave an exact quantum query algorithm to distinguish such strings. Now, we consider the D-J problem from the perspective of functional correlation. In fact, it is a problem equivalent to distinguishing $|C(f,\textbf{0})|=0$ from $|C(f,\textbf{0})|=1$, where \textbf{0} denotes that $f(x)$ is  0 for all values of $x$.\footnote{For the sake of beauty, here we abuse the absolute value notation.} Based on the above facts, we generalize the idea of Deutsch-Jozsa quantum algorithm to distinguish the correlation $|C(f,g)|=\varepsilon$ form $|C(f,g)|=1$ between two unknown functions.

\subsection{Exact quantum query algorithms for testing correlation}

\textbf{Problem 2} Correlation Testing Problem:
 Given two Boolean functions $f,g$ as black boxes, the testing  problem is to determine whether $|C(f,g)|=\varepsilon$  or $|C(f,g)|=1$, where $0\leq\varepsilon\leq\frac{\sqrt{3}}{2}$, under the promise that one of them holds.\vspace*{-1mm}

\begin{algorithm}\small
\caption{Quantum Algorithm for Correlation Testing}\label{alg3}
\KwIn{Black-boxes for $f$ and $g$}
\KwOut{ $|C(f,g)|=1$ iff $z=0$ }
$|\psi_{0}\rangle\leftarrow {|{0}\rangle^{\otimes n}\otimes |1\rangle}$\;
Apply ${I^{\otimes n}\otimes  H}$ to $|\psi_{0}\rangle$ and get
$|\psi_{1}\rangle=|{0}\rangle^{\otimes n}\otimes|-\rangle$\;
Apply ${H^{\otimes n}\otimes  I}$ to $|\psi_{1}\rangle$ and get
$|\psi_{2}\rangle=\sum_{x\in\{0,1\}^{n}}\frac{1}{\sqrt{2^{n}}}|x\rangle|-\rangle$\;
Apply ${U_{f}}$ to $|\psi_{2}\rangle$ and get $|\psi_{3}\rangle=\sum_{x\in\{0,1\}^{n}}\frac{{(-1)}^{f(x)}}{\sqrt{2^{n}}}|x\rangle|-\rangle$\;
Apply ${U_{g}}$ to $|\psi_{3}\rangle$ and get $|\psi_{4}\rangle=\sum_{x\in\{0,1\}^{n}}\frac{{(-1)}^{f(x)+g(x)}}{\sqrt{2^{n}}}|x\rangle|-\rangle$\;
Apply ${H^{\otimes n}\otimes  I}$ to $|\psi_{4}\rangle$ and get $|\psi_{5}\rangle= \sum_{z\in\{0,1\}^{n}}{W_{h}(z)}|z\rangle|-\rangle=|\Psi\rangle|-\rangle=|\Psi_{1}\rangle|-\rangle+|\Psi_{0}\rangle|-\rangle$, where $|\Psi_{1}\rangle=\sum_{z\neq0}{W_{h}(z)}|z\rangle $, $|\Psi_{0}\rangle=\sum_{z= 0}{W_{h}(z)}|z\rangle$ and $h=f+g$\;	
Apply $Q$ to $|\psi_{5}\rangle$ and get $|\psi_{final}\rangle$, where $ Q=-\mathcal {A}S_{0}(\phi)\mathcal{A}^{-1}S_{z}(\varphi), \mathcal {A}|0\rangle^{\otimes n}|-\rangle=|\Psi\rangle|-\rangle$, $\phi=2\arcsin\frac{1}{2\sqrt{1-\varepsilon^{2}}}, \varphi=2\arcsin\frac{1}{2\sqrt{1-\varepsilon^{2}}}$, $\mathcal {A}=(H^{\otimes n}\otimes  I)({U_{g}})({U_{f}})(H^{\otimes n}\otimes  I)$ and $S_{z}(\varphi)$ is defined such that
 \begin{displaymath}
|z\rangle|-\rangle\rightarrow \left\{ \begin{array}{ll}
e^{i\varphi}|z\rangle|-\rangle & \textrm{if $z\neq0,$}\\
\hspace{0.42cm}|z\rangle|-\rangle & \textrm{if $z=0.$}
\end{array} \right.
\end{displaymath}
Measurement the first $n$ qubits and get $z$\;
Return $z$\;
\end{algorithm}\vspace*{-3mm}

\begin{theorem}Algorithm \ref{alg3} solves Problem 2 with certainty using 6 queries only.
\end{theorem}
\begin{proof}
Case 1: $|C(f,g)|=1$. We obtain $\mathcal {A}|0\rangle^{\otimes n}|-\rangle=|\Psi\rangle|-\rangle=|\Psi_{0}\rangle|-\rangle$. Hence
\begin{align}\label{2}\notag
Q|\Psi\rangle|-\rangle=&Q|\Psi_{0}\rangle|-\rangle\\ \notag
=&-\mathcal {A}S_{0}(\phi)\mathcal{A}^{-1}S_{z}(\varphi)|\Psi_{0}\rangle|-\rangle\\
=&-\mathcal {A}(I-(1-e^{i\phi})|0\rangle^{\otimes n} \langle 0|^{\otimes n}\otimes |-\rangle \langle -|)\mathcal{A}^{-1}I|\Psi_{0}\rangle|-\rangle\\ \notag
=&-|\Psi_{0}\rangle|-\rangle+(1-e^{i\phi})|\Psi_{0}\rangle|-\rangle\\ \notag
=&-e^{i\phi}|\Psi_{0}\rangle|-\rangle. \notag
\end{align}
Therefore, the result of measurement is $z=0$, and Algorithm \ref{alg3} outputs ``$|C(f)|=1$" with certainty.

\medskip
Case 2a: $|C(f,g)|=\varepsilon=0$. We obtain $\mathcal {A}|{0}\rangle^{\otimes n}|-\rangle=|\Psi\rangle|-\rangle=|\Psi_{1}\rangle|-\rangle$. Hence
\begin{align}\label{2} \notag
Q|\Psi\rangle|-\rangle=&Q|\Psi_{1}\rangle|-\rangle\\ \notag
=&-\mathcal {A}S_{0}(\phi)\mathcal{A}^{-1}S_{z}(\varphi)|\Psi_{1}\rangle|-\rangle\\
=&-\mathcal {A}(I-(1-e^{i\phi})|0\rangle^{\otimes n} \langle 0|^{\otimes n}\otimes |-\rangle \langle -|)\mathcal{A}^{-1}e^{i\varphi}|\Psi_{1}\rangle|-\rangle\\ \notag
=&e^{i\varphi}(-|\Psi_{1}\rangle+(1-e^{i\phi})|\Psi_{1}\rangle)|-\rangle\\ \notag
=&-e^{i(\varphi+\phi)}|\Psi_{1}\rangle|-\rangle. \notag
\end{align}
Therefore, the result of measurement is $z\neq0$, and Algorithm \ref{alg3} outputs ``$|C(f,g)|=\varepsilon$" with certainty.

\medskip
Case 2b: $|C(f,g)|=\varepsilon\in(0,\frac{\sqrt{3}}{2})$. This is equivalent to $|W_{h}(0)|=\varepsilon$, so $\sin^{2}(\theta_{a})=a=\langle\Psi_{1}|\Psi_{1}\rangle=1-\varepsilon^{2}$. Finally, we have $0<\pi/4\theta_{a}-\frac{1}{2}<1$ and $m=\lfloor {\pi/4\theta_{a}}-\frac{1}{2}\rfloor=0.$ Therefore, we only need to compute the parameters value $\varphi$ and $\phi$ using $Q$ once in order to achieve a precise distinction.

\medskip
Let us now choose $\varphi$ and $\phi$ such that
\begin{align}\label{2}
e^{i\varphi}(1-e^{i\phi})\sqrt{a}\sin(\theta_{a})=((1-e^{i\phi})a+e^{i\phi})\frac{1}{\sqrt{1-a}}\cos(\theta_{a}).
\end{align}
This equation is equivalent to
\begin{align}\label{2}
(1-\varepsilon^{2})(1-e^{i\phi})(e^{i\varphi}-1)=e^{i\phi}.
\end{align}
and by solving this equation we obtain $\phi=\varphi=2\arcsin\frac{1}{2\sqrt{1-\varepsilon^{2}}}$.
By Lemma \ref{Generalized Amplitude Expansion}, we get
\begin{align}\label{26}
|\Psi_{final}\rangle=Q|\Psi\rangle|-\rangle=\left[a(e^{i\varphi}-1)(1-e^{i\phi})-(e^{i(\phi)}+e^{i(\varphi)})+1\right]|\Psi_{1}\rangle|-\rangle.
\end{align}
Therefore, the result of the measurement provides $z\neq0$, i.e., Algorithm \ref{alg3} outputs ``$|C(f,g)|=\varepsilon$" with certainty.

\medskip
Case 2c: $|C(f,g)|=\varepsilon=\frac{\sqrt{3}}{2}$. This is equivalent to $|W_{h}(0)|=\frac{\sqrt{3}}{2}$, so $\sin^{2}(\theta_{a})=a=\langle\Psi_{1}|\Psi_{1}\rangle=\frac{1}{4}$. We have $\phi=\varphi=2\arcsin\frac{1}{2\sqrt{1-\varepsilon^{2}}}=\pi$.
By Lemma \ref{Generalized Amplitude Expansion}, the amplitude of state $|\Psi_{0}\rangle|-\rangle$ is
\begin{align}\label{6}
&e^{i\varphi}(1-e^{i\phi})a-((1-e^{i\phi})a+e^{i\phi})\\
&= -2a(1-e^{i\phi})-e^{i\phi}=0.
\end{align} That is, equation \ref{26} holds. Therefore, the result of the measurement provides $z\neq0$, i.e., Algorithm \ref{alg3} outputs ``$|C(f,g)|=\varepsilon$" with certainty.
\end{proof}
\begin{theorem}
For correlation testing Problem 2, the classical deterministic query complexity is $\Theta(N)$. More specifically, solving Problem 2 with zero error requires $(1-\varepsilon)N+2$ or $(1+\varepsilon)N+2$ queries.
\end{theorem}
\begin{proof}
 If the first $(\frac{1-\varepsilon}{2})N$ input queries return $f(x_{i})\neq g(x_{i})$ or $(\frac{1+\varepsilon}{2})N$ input queries return $f(x_{i})=g(x_{i})$, then we will need to make another query as well. So the total number of queries is $(1-\varepsilon)N+2$  or $(1+\varepsilon)N+2$.
Therefore the theorem holds.
\end{proof}

\section{Testing balancedness of Boolean functions }

Balancedness is also one of the properties of Boolean functions important in cryptography, see for example \cite{IEEE16,DCC18}. Of interest are also functions that are not balanced, but $\varepsilon$-far balanced. To this end, we introduce the following concept. We study the balance problem from the perspective of function property testing.
In the following we study balancedness and $\varepsilon$-far balancedness of Boolean functions.
\begin{definition}
A Boolean function $f(x)$ is called $\varepsilon$-far balanced if $|C(f)|\geq\varepsilon$.
\end{definition}
\textbf{Problem 3} Balancedness Testing Problem:
 Given an unknown Boolean functions $f$, balancedness testing  problem is to determine whether $f$ is balanced or is $\varepsilon$-far balanced, under the promise that one of them holds.
\begin{algorithm}\small
\caption{Quantum Algorithm for Balancedness Testing}\label{alg4}
\KwIn{Black-boxes for $f$ and parameter $\varepsilon>0$}
\KwOut{ $f$ is balanced iff $z\neq0$ }
$|\psi_{0}\rangle\leftarrow {|{0}\rangle^{\otimes n}\otimes|1\rangle}$\;
Apply ${I^{\otimes n}\otimes  H}$ to $|\psi_{0}\rangle$ and get
$|\psi_{1}\rangle=|{0}\rangle^{\otimes n}\otimes|-\rangle$\;
Apply ${H^{\otimes n}\otimes  I}$ to $|\psi_{1}\rangle$ and get
$|\psi_{2}\rangle=\sum_{x\in\{0,1\}^{n}}\frac{1}{\sqrt{2^{n}}}|x\rangle|-\rangle$\;
Apply ${U_{f}}$ to $|\psi_{2}\rangle$ and get $|\psi_{3}\rangle=\sum_{x\in\{0,1\}^{n}}\frac{{(-1)}^{f(x)}}{\sqrt{2^{n}}}|x\rangle|-\rangle$\;
Apply ${H^{\otimes n}\otimes  I}$ to $|\psi_{3}\rangle$ and get $|\psi_{4}\rangle= \sum_{z\in\{0,1\}^{n}}{W_{f}(z)}|z\rangle|-\rangle=|\Psi\rangle|-\rangle=|\Psi_{1}\rangle|-\rangle+|\Psi_{0}\rangle|-\rangle$, where $|\Psi_{1}\rangle=\sum_{z=0}{W_{f}(z)}|z\rangle $, $|\Psi_{0}\rangle=\sum_{z\neq 0}{W_{f}(z)}|z\rangle$ \;	
Apply $Q^{m}$ to $|\psi_{4}\rangle$ and get $|\psi_{final}\rangle$, where $Q=-\mathcal {A}S_{0}\mathcal{A}^{-1}S_{\chi}$, $m=O(\frac{1}{\varepsilon})$, $\mathcal {A}|{0}\rangle^{\otimes n}|-\rangle=|\Psi\rangle|-\rangle$, $\mathcal {A}=(H^{\otimes n}\otimes  I)({U_{f}})(H^{\otimes n}\otimes  I)$  and $|0\rangle|-\rangle$ is the only good state\;
Measurement the first $n$ qubits and get $z$\;
Return $z$\;
\end{algorithm}

\begin{theorem}
 Algorithm \ref{alg4} solves Balancedness Testing Problem 3 with one-sided error using $O(\frac{1}{\varepsilon})$ queries. If $f$ is balanced, Algorithm \ref{alg4} outputs ``$f$ is balanced" with certainty. If $f$ is $\varepsilon$-far balanced, Algorithm \ref{alg4} outputs that ``$f$ is $\varepsilon$-far balanced" with probability at least $\frac{3}{4}$.
\end{theorem}
\begin{proof}
Case 1: $f$ is balanced, i.e., $|C(f)|=0$. We obtain $\mathcal {A}|{0}\rangle^{\otimes n}|-\rangle=|\Psi\rangle|-\rangle=|\Psi_{0}\rangle|-\rangle$. Hence
\begin{align}\label{2} \notag
Q|\Psi\rangle|-\rangle=&Q|\Psi_{0}\rangle|-\rangle\\ \notag
=&-\mathcal {A}S_{0}\mathcal{A}^{-1}S_{\chi}|\Psi_{0}\rangle|-\rangle\\
=&-\mathcal {A}(I-2|0\rangle^{\otimes n} \langle 0|^{\otimes n}\otimes |-\rangle \langle -|)\mathcal{A}^{-1}I|\Psi_{0}\rangle|-\rangle\\ \notag
=&-|\Psi_{0}\rangle|-\rangle+2|\Psi_{0}\rangle|-\rangle\\ \notag
=&|\Psi_{0}\rangle|-\rangle=|\Psi\rangle|-\rangle. \notag
\end{align}
Therefore, Algorithm \ref{alg4} outputs ``$f$ is balanced" with certainty.

\medskip
Case 2: $f$ is $\varepsilon$-far balanced, i.e., $|C(f)|\geq\varepsilon$, $|W_{f}(0)|=|C(f)|\geq\varepsilon$. In this case the final state $|\psi_{final}\rangle$ in Algorithm \ref{alg4} is:
\begin{align}\label{2}
\sum_{z\in\{0,1\}^{n}}{W_{f}(z)}|z\rangle|-\rangle=&|\Psi_{1}\rangle|-\rangle+|\Psi_{0}\rangle|-\rangle ,\\
|\Psi_{1}\rangle=&\sum_{z= 0}{W_{f}(z)}|z\rangle,\\
|\Psi_{0}\rangle=&\sum_{z\neq0}{W_{f}(z)}|z\rangle,
\end{align} and therefore
\begin{align}\label{3}
a &=\rm sin^{2}\theta_{a}=\langle\Psi_{1}\mid\Psi_{1}\rangle\\ \notag
  &=C^{2}(f)\geq\varepsilon^{2}.
\end{align}
 Hence, $\theta_{a}\geq \arcsin{\varepsilon}$.

\medskip
In order to prepare the final state $|\psi_{final}\rangle$ in Algorithm \ref{alg4}, we need to call the function black-box two times, i.e., to use two queries. So it takes in total 4 queries (indeed, $\mathcal{A}$ uses two queries, $\mathcal{A}^{-1}$ uses  two queries also) when Algorithm \ref{alg4} runs $Q$ once, where $Q=-\mathcal {A}S_{0}\mathcal{A}^{-1}S_{\chi}$. According to Lemma \ref{lm1}, Algorithm \ref{alg4} runs $Q$  $m=\lfloor {\pi/4\theta_{a}}\rfloor$ times. Since the total number of queries of Algorithm \ref{alg4} is $2m+1=O(\frac{1}{\varepsilon})$, i.e., the query complexity of Algorithm \ref{alg4} is $O(\frac{1}{\varepsilon})$. Without loss of generality, let's now assume $\varepsilon\leq 1/2 $. In such a case we get, using Lemma \ref{lm1}, Algorithm \ref{alg4} outputs  ``$f$ is $\varepsilon$-far balanced" with the probability $p\geq \rm max(1-1/4,1/4)>\frac{3}{4}$.
\end{proof}

Finally, we prove that our quantum algorithm for Problem 4 is optimal.\\
For $X=(x_{1},\ldots,x_{n})$$\in\{0,1\}^{n}$. Let $l,l'$ be integers such that $0\leq l,l'\leq n$. Let us now define the partial Boolean function $f_{l,l'}$ on$\{0,1\}^{n}$, as
\begin{align}
	f_{l,l'}(X)= \begin{cases}
		   1 & \mbox{if~} |X|=l, \\
           0 & \mbox{if~} |X|= l'\\
           \mbox {undefined}  & \mbox{otherwise.}
	\end{cases}
\end{align}
In addition, let $m\in\{l,l'\}$ be such that $|\frac{n}{2}-m|$is maximized, and let $\Delta_{l}=|l-l'|$.
\begin{lemma}(\cite{arXiv98})\label{lower bound3}
A lower bound for quantum query complexity of computing the partial function $f_{l,l'}$, given the input as an oracle, is $\Omega\left(\sqrt{n/\Delta_{l}}+\sqrt{m(n-m)}/\Delta_{l}\right)$.
\end{lemma}
\begin{theorem}\label{qlower bound-bt}
Any quantum query algorithm for Balancedness Testing Problem 3 requires $\Omega(\frac{1}{\varepsilon})$ queries.
\end{theorem}
\begin{proof}
 Let $Y=(y_{0},y_{1},\ldots,y_{N-1})$ where $N=2^n$, for $i=0,1,\cdots, N-1$, $y_{i}\in \{0,1\}$ and  $f(i) = y_i$. The value of any $n$-bit Boolean function can be expressed by an $N$-bit string $\,Y$. The task$\,$ to solve the$\,$~Balancedenes$\,$ problem$\,$  can$\,$ now be reduced$\,$ to the problem  to distinguish the$\,$ Hamming$\,$ weight $|Y|=\frac {N}{2}\,$ from $\,|Y|\leq\frac {(1-\varepsilon)N}{2 }\,$ or $\,|Y|\geq\frac {(1+\varepsilon)N}{2 }\,$. In view of this, let us consider$\,$ two arbitrary $\,N$-bit

 \eject

 \noindent   symmetric  partial  Boolean functions $H_{1},H_{2}: \{D\subset\{0,1\}^{N}\}\to\{0,1\}$ (for a subset $D$), which are defined as follows:
\begin{align}
	H_{1}(Y)= \begin{cases}
		   1& \mbox{if~} |Y|=\frac {N}{2}, \\
           0 & \mbox{if~} |Y|=\frac {(1-\varepsilon)N}{2 }\hspace{0.05cm}\\
           \mbox {undefined}  & \mbox{otherwise.}
	\end{cases}
  \\
	H_{2}(Y)= \begin{cases}
		   1& \mbox{if~} |Y|=\frac {N}{2}, \\
           0 & \mbox{if~} |Y|=\frac {(1+\varepsilon)}{2}N, \\
           \mbox {undefined}  & \mbox{otherwise.}
	\end{cases}
\end{align}
We consider here only the case of calculating $H_{1}(Y)$. The other case is similar. Let $l'=\frac {(1-\varepsilon)}{2}N$, $l=\frac {N}{2}$. We take $m=\frac {(1-\varepsilon)}{2}N$ such that $|\frac{N}{2}-m|$ is maximized, and let $\Delta_{l}=|l-l'|=\frac {\varepsilon}{2}N$. Therefore we have
\begin{align}
&\Omega\left(\sqrt{N/\Delta_{l}}+\sqrt{m(N-m)}/\Delta_{l}\right)\\[2pt]
&=\Omega\left(\sqrt{N/\frac {\varepsilon}{2}N}+\sqrt{\frac {1-\varepsilon}{2}N(N-\frac {1-\varepsilon}{2}N)}/\frac {\varepsilon}{2}N\right)\\[2pt]
&=\Omega\left(\sqrt{2/\varepsilon}+\sqrt{1-\varepsilon^{2}}/\varepsilon\right)\\
&=\Omega\left(1/\varepsilon\right).
\end{align}
According to Lemma \ref{lower bound3}, the query complexity of computing  $H_{1}(Y)$( or $H_{2}(Y)$) is $\Omega(1/\varepsilon)$. Therefore, the theorem holds.
\end{proof}
Therefore, the quantum query complexity of Balancedness Testing Problem 3 is $\Theta(\frac{1}{\varepsilon})$. Next, let's consider the classical algorithm for this problem. We first give the de Moivre-Laplace Central Limit Theorem as follows:
\begin{lemma}(\cite{Miller2017})\label{Central Limit Theorem}(de Moivre-Laplace Central Limit Theorem).
If $X$ is a random variable having the binomial with the parameters $n$ and $p$, the limiting form  of the distribution function of the standardized random variable
\begin{displaymath}
Z=\frac {X-np}{\sqrt{np(1-p)}}
\end{displaymath}as $n\rightarrow\infty$, is given by the standard normal distribution
\begin{displaymath}
F(z)=\int_{-\infty}^{z}\frac {1}{\sqrt{2\pi}}e^{-\frac {t^{2}}{2}}dt, -\infty<z<\infty.
\end{displaymath}
\end{lemma}Using the equation $P(a<Z\leq b)=F(b)-F(a)$, the de Moivre-Laplace Central Limit Theorem has the following equivalent forms
\begin{align}\label{Central Limit Theorem-eq}
\lim_{n\rightarrow\infty}P[a\leq \frac { X-np}{\sqrt{np(1-p)}}\leq b]=\int_{a}^{b}\frac {1}{\sqrt{2\pi}}e^{-\frac {t^{2}}{2}}dt.
\end{align}
\begin{algorithm}\small
\caption{Classical Randomized Algorithm for Balancedness Testing}\label{alg5}
\KwIn{Black-boxes for $f$ and parameter $\varepsilon>0$}
\KwOut{ $f$ is balanced iff $|C'(f)|\leq0.5\varepsilon$ }
$T\leftarrow O(1/{\varepsilon}^{2})$\;
Take $T$ elements $\{x_{1},x_{2},\cdots,x_{T}\}$ uniformly at random, where $x_{i}\in\{0,1\}^{n}$\;
Compute $|C'(f)|=|\Sigma_{i=1}^{T}(-1)^{f(x_{i})}|$\;
Return $|C'(f)|$\;
\end{algorithm}\vspace*{-1mm}

\begin{theorem}\label{Cubound-bt}
 Algorithm \ref{alg5} solves Balancedness Testing Problem 3 using $O( 1/{\varepsilon^{2}})$ queries. If $f$ is balanced, Algorithm \ref{alg5} outputs ``$f$ is balanced" with the probability at least $\frac{2}{3}$; If $f$ is $\varepsilon$-far balanced, Algorithm \ref{alg5} outputs ``$f$ is $\varepsilon$-far balanced" with the probability at least $\frac{2}{3}$.
\end{theorem}
\begin{proof}
Let $A_f=\{x:f(x)=1\}$. Let us now consider an experiment in which we will choose $T$ times independently and randomly elements $x_{i}\in \{0,1\}^{n}$.
 Let $\xi$ be the number of times that those chosen elements occur in $A_f$ in T trials, i.e., $\xi=|\{x_{i}:f(x_{i})=1\}|$. Then $\xi$ is a random variable. Moreover, we have
$\xi\sim b(T,c)$, where $P(A_f)=c$  and $b(T,c)$ denotes binomial distribution. The distribution law for $\xi$ is as follows:
$P\{\xi=k\}=\left(
                                             \begin{array}{c}
                                               T \\
                                               k \\
                                             \end{array}
                                           \right)
$$c^{k}(1-c)^{T-k}, k=0,1,\cdots,T$. So the probability is
\begin{align}\label{1}
&P\{(c-\delta)T\leq\xi\leq(c+\delta)T\}=\sum_{k=(c-\delta)T}^{(c+\delta)T}\left(
                                             \begin{array}{c}
                                               T \\
                                               k \\
                                             \end{array}
                                           \right)
c^{k}(1-c)^{T-k},\end{align} where $\delta$ is an accuracy parameter.

\medskip
It is difficult to find the exact value for the probability we need. We use The de Moivre-Laplace Central Limit Theorem to find its approximation.
Let\begin{align}\label{1}
&P\left\{(c-\delta)T\leq\xi\leq(c+\delta)T\right\}\\
=&P\left\{\frac{(c-\delta)T-Tc}{\sqrt{Tc(1-c)}}\leq \frac{\xi-Tc}{\sqrt{Tc(1-c)}} \leq\frac{(c+\delta)T-Tc}{\sqrt{Tc(1-c)}}\right\}\\
\approx&\int_{\frac{(c-\delta)T-Tc}{\sqrt{Tc(1-c)}}}^{\frac{(c+\delta)T-Tc}{\sqrt{Tc(1-c)}}}\frac{1}{\sqrt{2\pi}}e^{-\frac{t^{2}}{2}}dt\hspace{0.1cm} ( \mbox{Using}\hspace{0.1cm} \mbox{Equation} \hspace{0.1cm}\ref{Central Limit Theorem-eq}\hspace{0.1cm}\mbox{approximation})\\
=&\int_{\frac{-\delta T}{\sqrt{Tc(1-c)}}}^{\frac{\delta T}{\sqrt{Tc(1-c)}}}\frac{1}{\sqrt{2\pi}}e^{-\frac{t^{2}}{2}}dt\\
=&F\left({\frac{\delta T}{\sqrt{Tc(1-c)}}}\right)-F\left({\frac{-\delta T}{\sqrt{Tc(1-c)}}}\right)\\
=&2F\left({\frac{\delta T}{\sqrt{Tc(1-c)}}}\right)-1\geq\frac{2}{3}.
\end{align} By checking the value of the distribution function, we have ${\frac{\delta T}{\sqrt{Tc(1-c)}}}\geq 1$, i.e. $T=\Omega( 1/{\delta^{2}})$.

\medskip
Case 1: $f$ is balanced, i.e., $|C(f)|=0$. That is $c=0.5$. If the random variable $\xi$ satisfies the inequality $(c-\delta)T\leq\xi\leq(c+\delta)T$, that is $0.5-\delta\leq\frac{\xi}{T}\leq0.5+\delta$, we know that $-2\delta\leq1-\frac{2\xi}{T}\leq 2\delta$. Hence,
\begin{align}\label{2}
|C'(f)|=&\frac{1}{T}\left|\sum_{x_{i}}{(-1)}^{f(x_{i})}\right|\\
       =&\frac{1}{T}\left|\sum_{x_{i}:f(x_{i})=1}{(-1)}^{f(x_{i})}+\sum_{x_{i}:f(x_{i})=0}{(-1)}^{f(x_{i})}\right|\\
       =&\frac{1}{T}\left|-\xi+T-\xi\right|=\left|1-\frac{2\xi}{T}\right|\leq2\delta.
\end{align}
If we now take $\delta=0.1\varepsilon$, then $P\left\{(0.5-0.1\varepsilon)T\leq\xi\leq(0.5+0.1\varepsilon)T\right\}\geq\frac{2}{3}$ using $T=O( 1/{\delta^{2}})=O( 1/{\varepsilon^{2}})$ queries. In other words we get that $|C'(f)|\leq0.2\varepsilon$ with the probability at least $\frac{2}{3}$. Naturally, $|C'(f)|\leq0.5\varepsilon$ with the probability at least $\frac{2}{3}$. That is Algorithm \ref{alg5} outputs $|C(f)|=0$
with the probability at least $\frac{2}{3}$ using $T=O( 1/{\delta^{2}})=O( 1/{\varepsilon^{2}})$ queries.

\medskip
Case 2: $f$ is $\varepsilon$-far balanced, i.e., $|C(f)|\geq\varepsilon$. Let's consider the case of $|C(f)|=\varepsilon$.  That is $c=\frac{1\pm\varepsilon}{2}$. In order to discuss it we need to consider two cases.

\medskip
Case 2a: $c=\frac{1+\varepsilon}{2}$. If the random variable $\xi$ satisfies the inequality $(c-\delta)T\leq\xi\leq(c+\delta)T$, that is $\frac{1+\varepsilon}{2}-\delta\leq\frac{\xi}{T}\leq\frac{1+\varepsilon}{2}+\delta$. Hence, we have $-\varepsilon-2\delta\leq1-\frac{2\xi}{T}\leq-\varepsilon+ 2\delta$. If we now take $\delta=0.1\varepsilon$, then $P\left\{(\frac{1+\varepsilon}{2}-0.1\varepsilon)T\leq\xi\leq(\frac{1+\varepsilon}{2}+0.1\varepsilon)T\right\}\geq\frac{2}{3}$ using $T=O( 1/{\delta^{2}})=O( 1/{\varepsilon^{2}})$ queries. In other words, we can draw a conclusion that $|C'(f)|=|1-\frac{2\xi}{T}|\geq0.8\varepsilon>0.5\varepsilon$ with the probability at least $\frac{2}{3}$ using $O( 1/{\varepsilon^{2}})$ queries.
That means that Algorithm \ref{alg5} outputs $|C(f)|\geq\varepsilon$ with probability at least $\frac{2}{3}$ using $O( 1/{\varepsilon^{2}})$ queries.

\medskip
Case 2b: $c=\frac{1-\varepsilon}{2}$. If the random variable $\xi$ satisfies the inequality $(c-\delta)T\leq\xi\leq(c+\delta)T$, that is  $\frac{1-\varepsilon}{2}-\delta\leq\frac{\xi}{T}\leq\frac{1-\varepsilon}{2}+\delta$. Hence, we know that $\varepsilon-2\delta\leq1-\frac{2\xi}{T}\leq \varepsilon+ 2\delta$. If we now take $\delta=0.1\varepsilon$, then $P\left\{(\frac{1-\varepsilon}{2}-0.1\varepsilon)T\leq\xi\leq(\frac{1-\varepsilon}{2}+0.1\varepsilon)T\right\}\geq\frac{2}{3}$ using $T=O( 1/{\delta^{2}})=O( 1/{\varepsilon^{2}})$ queries. In other words, we get that $|C'(f)|=|1-\frac{2\xi}{T}|\geq0.8\varepsilon>0.5\varepsilon$ with the probability at least $\frac{2}{3}$ using $O( 1/{\varepsilon^{2}})$ queries.
That is Algorithm \ref{alg5} outputs in any case $|C(f)|\geq\varepsilon$ with the probability at least $\frac{2}{3}$ using $O( 1/{\varepsilon^{2}})$ queries.
\end{proof}

Next, we prove optimality of our Algorithm \ref{alg5}.
\begin{lemma}\label{iid}(\cite{Chernoff10})
Let $X_{1},X_{2},\cdots,X_{m}$ be i.i.d random variables taking 0 or 1, and $\mbox {Pr}[X_{i}=1]=p$.

 If $p\leq\frac{1}{4}$, then for any $t\geq0$
 \begin{align}
 \mbox {Pr}\left[(\sum_{i=1}^{m}X_{i}-\mu)>t\right]\geq\frac{1}{4}\mbox{exp}(\frac{-2t^{2}}{\mu}).
 \end{align}

 If $p\leq\frac{1}{2}$, then for any $0\leq t\leq m(1-2p)$
 \begin{align}
 \mbox {Pr}\left[(\sum_{i=1}^{m}X_{i}-\mu)>t\right]\geq\frac{1}{4}\mbox{exp}(\frac{-2t^{2}}{\mu}),
 \end{align} where $\mu=E[\sum_{i=1}^{m}X_{i}]=mp$.

\end{lemma}
\begin{theorem}\label{CLbound-bt}
Any classical randomized algorithm for Balancedness Testing Problem 3 requires $\Omega(\frac{1}{\varepsilon^{2}})$ queries.
\end{theorem}
\begin{proof}
Let $Y=(y_{1},y_{2},\ldots,y_{N})$. Consider a computation of the partial function $H_{1}(Y)$ from the proof of Theorem \ref{qlower bound-bt}. Select $m$ strings $y'_{1}, y'_{2},\cdots,y'_{m} $ in $Y$ independently and randomly. Let $p=\frac{1-\varepsilon}{2}$ and $t=\frac{m\varepsilon}{2}$. By Lemma \ref{iid}, we obtain
 \begin{align} \notag
 &\mbox {Pr}\left[(\sum_{i=1}^{m}y'_{i}-\mu)>t\right]\\ \notag
 =&\mbox {Pr}\left[(\sum_{i=1}^{m}y'_{i}-m\frac{1-\varepsilon}{2})>\frac{m\varepsilon}{2}\right]\\
 =&\mbox {Pr}\left[\sum_{i=1}^{m}y'_{i}>\frac{m}{2}\right]\\ \notag
\geq&\frac{1}{4}\mbox{exp}\left(\frac{-2t^{2}}{\mu}\right)\\ \notag
=&\frac{1}{4}\mbox{exp}\left(\frac{-m\varepsilon^{2}}{1-\varepsilon}\right). \notag
 \end{align}
If $m=o(\frac{1}{\varepsilon^{2}})$, then $\mbox {Pr}\left[\sum_{i=1}^{m}y'_{i}>\frac{m}{2}\right]\geq\frac{1}{4}$. At this point, we cannot tell which case it is. Repeating a constant number of experiments, the error probability can reach $\frac{1}{3}$. Therefore, $m=\Omega(\frac{1}{\varepsilon^{2}})$.
\end{proof}
\begin{remark}
According to Theorem \ref{Cubound-bt} and Theorem \ref{CLbound-bt}, the randomized query complexity of our Problem 3 is $\Theta(\frac{1}{\varepsilon^{2}})$.
\end{remark}

\section{Conclusions}\label{Sec6}

In this paper, we have given classical and quantum algorithms for the problem of testing the identity of Boolean functions, which  may be thought of as  a special case of testing the isomorphism and the affine equivalence of Boolean functions. At the same time, we have also proved optimality of presented algorithms. By a combination of those results, we have obtained an optimal separation in the query complexities of the following identity testing problems: $\Theta(\frac{1}{\varepsilon})$ versus $\Theta(\frac{1}{\sqrt{\varepsilon}})$. In addition, we have extend the idea of the D-J algorithm to the balancedness testing and correlation testing problems.

Our results naturally raise a variety of related open problems for future explorations.  For example, what are  properties of isomorphic Boolean functions and affine equivalent Boolean functions? In particular, what is  the (optimal) quantum query complexity  for  affine equivalent Boolean functions? In more general cases, given two unknown Boolean functions $f,g$, the testing  problem is to determine whether $|C(f,g)|\leq a$ or $|C(f,g)|\geq b$, where $0<\varepsilon\leq b-a\leq1$, under the promise that one of these cases holds. Obviously, we have already solved some special cases ($\varepsilon$ should be a relatively small number, see the following Table 1 for more details, where CQC means Classical query complexity  and QQC means Quantum query complexity). We found that the query complexity of quantum or classic algorithms  with an error depends not only on the distance between $b$ and $a$, but also on the location of $b$ and $a$  and is independent of $n$. The question is, for example, when $b-a=\varepsilon$, the quantum query complexities of the Identity testing problem and the Balancedness testing problem are $\Theta(\frac{1}{\sqrt{\varepsilon}})$ and $\Theta(\frac{1}{\varepsilon})$ respectively. At the same time, when $b-a=\varepsilon$, the classical query complexities of the above two questions are  $\Theta(\frac{1}{\varepsilon})$ and  $\Theta(\frac{1}{\varepsilon^{2}})$ respectively. Besides, although the expression of quantum or classical query complexity function is different in above two problems, the separation between quantum and classical algorithms is same in these two special cases. Whether the same separation properties between quantum and classical algorithms can be generalized to general situations (i.e., for different $a,b$ values) remains to be further investigated.

\subsection*{Acknowledgements}

The authors would like to thank the referees and Professor Calude for important comments that help us improve the quality of the manuscript.
This work  is  supported in part by the National Natural Science Foundation of China (Nos. 61572532,  61876195), and the Natural Science Foundation of Guangdong Province of China (No. 2017B030311011).

\begin{table}[!htbp]
\caption{Query complexity results}
\centering
\begin{tabular}{|l|c|l|c|c|}
\hline
Problem& Algorithm&Parameter values & QQC& CQC\\
\hline
Identity testing&Bounded Error&$a=1-2\varepsilon,b=1$ & $\Theta(1/{\sqrt{\varepsilon}})$& $\Theta(1/{\varepsilon})$\\
\hline
 Deutsch-Jozsa Problem&Exact& $a=0, b=1$& $O(1)$ & $\Theta(N)$\\
\hline
Correlation testing&Exact& $a=\varepsilon, b=1$& $O(1)$ & $\Theta(N)$\\
\hline
Balancedness testing&Bounded Error&$a=0,b=\varepsilon$  & $\Theta(1/{\varepsilon})$ & $\Theta(1/{\varepsilon^{2}})$\\
\hline
\end{tabular}
\end{table}


\end{document}